\documentclass[conference]{IEEEtran}
\IEEEoverridecommandlockouts
\usepackage{cite}
\usepackage{amsmath,amssymb,amsfonts,amsthm}
\usepackage{dsfont}
\usepackage{algorithmic}
\usepackage{graphicx}
\usepackage{textcomp}
\usepackage{xcolor}
\def\BibTeX{{\rm B\kern-.05em{\sc i\kern-.025em b}\kern-.08em
    T\kern-.1667em\lower.7ex\hbox{E}\kern-.125emX}}

\newtheorem{theorem}{Theorem}
\newtheorem{definition}[theorem]{Definition}
\newtheorem{lemma}[theorem]{Lemma}

\newcommand{\sfTr}[1]{\mathrm{Tr}\left\{#1\right\}}

\def\sfB{\mathsf{B}}

\def\bbH{\mathbb{H}}
\def\bbS{\mathbb{S}}
\def\bbL{\mathbb{L}}

\def \be {\begin{equation}}
\def \ee {\end{equation}}

\def \Re{\mathrm{Re}\,}
\def \Im{\mathrm{Im}\,}

\def \e{\mathrm{e}}
\def \m{\mathrm{m}}
\def \e1{(\mathrm{e})}
\def \m1{(\mathrm{m})}

\def \post{\mathrm{post}}

\newcommand{\lamin}[2]{\langle#1,#2\rangle_{\rhoB}^{(\lambda)}}

\def \cH{{\cal H}}

\def \cX{{\cal X}}

\def \bbc{{\mathbb C}}

\def \sofc2{{\cal S}({\mathbb C}^2)}

\def \prior{\pi(\theta)}

\def \htheta{\hat{\theta}}

\def \cNH{{\cal C}_\mathrm{BNH}}

\def \rhoB{{S}_{\mathrm{B}}}
\def \MBj{{D}_{\mathrm{B},j}}
\def \MBk{{D}_{\mathrm{B},k}}
\def \Mtheta{D(\theta)}

\def \sfZB{{Z}_{\mathrm{B}}}
\def \sfHB{{H}_{\mathrm{B}}}

\def \bbSbar{\overline{\bbS}}
\def \Mbar{\overline{D}}
\def \Mbarj{\overline{D}_j}
\def \sfwbar{\overline{w}}

\newcommand{\R}{\mathds{R}}
\newcommand{\C}{\mathds{C}}

\newcommand{\Li}{\mathcal{L}}

\newcommand{\Ccal}{\mathcal{C}}

\newcommand{\Sbb}{\mathbb{S}}
\newcommand{\Xds}{\mathds{X}}

\newcommand{\half}{\frac{1}{2}}

\newcommand{\Real}{\mathrm{Re}\,}
\newcommand{\Imag}{\mathrm{Im}\,}
\newcommand{\Hil}{\mathcal{H}}

\newcommand{\E}{\mathcal{E}}

\newcommand{\plam}{\frac{1+\lambda}{2}}
\newcommand{\mlam}{\frac{1-\lambda}{2}}

\newcommand{\hattheta}{\hat{\theta}}
\newcommand{\ta}{\theta}
\newcommand{\la}{\lambda}
\newcommand{\sq}{\sqrt}
\newcommand{\ot}{\otimes}

\newcommand{\braket}[2]{\langle #1|#2 \rangle}
\newcommand{\Tr}{\mathrm{Tr}}

\newcommand{\tr}{\mathrm{tr}}
\newcommand{\Ttr}{\mathds{T}\mathrm{r}}
\newcommand{\inner}[2]{\langle #1,#2 \rangle}
\newcommand{\inv}{{-1}}
\newcommand{\ol}{\overline}

\begin{document}

\title{Bayesian Logarithmic Derivative Type Lower Bounds for Quantum Estimation
\thanks{This work was supported in part by JSPS KAKENHI Grant Numbers JP21K11749, JP24K14816.}
}

\author{
\IEEEauthorblockN{1\textsuperscript{st} Jianchao Zhang}
\IEEEauthorblockA{\textit{Graduate School of Informatics and Engineering,} \\
\textit{The University of Electro-Communications, }\\
Tokyo, 182-8585 Japan \\
c2141016@edu.cc.uec.ac.jp}
\and
\IEEEauthorblockN{2\textsuperscript{nd} Jun Suzuki}
\IEEEauthorblockA{\textit{Graduate School of Informatics and Engineering,} \\
\textit{The University of Electro-Communications, }\\
Tokyo, 182-8585 Japan \\
junsuzuki@uec.ac.jp}

}

\maketitle

\begin{abstract}
Bayesian approach for quantum parameter estimation has gained a renewed interest from practical applications of quantum estimation theory. Recently, a lower bound, called the Bayesian Nagaoka-Hayashi bound for the Bayes risk in quantum domain was proposed, which is an extension of a new approach to point estimation of quantum states by Conlon {\it et al.} (2021). The objective of this paper is to explore this Bayesian Nagaoka-Hayashi bound further by obtaining its lower bounds. We first obtain one-parameter family of lower bounds, which is an analogue of the Holevo bound in point estimation. Thereby, we derive one-parameter family of Bayesian logarithmic derivative type lower bounds in a closed form for the parameter independent weight matrix setting. This new bound includes previously known Bayesian lower bounds as special cases.  

\end{abstract}


\section{Introduction}

Quantum estimation are pervasive in real-world scenarios, ranging from tracking moving objects to implementing quantum sensors \cite{proctor2018multiparameter,dinani2019bayesian}. These problems often involve multiple unknown parameters, such as determining the velocity of an object or quantifying phase information \cite{zhuang2017entanglement,szczykulska2017reaching}. In tackling these challenges, existing literature has extensively relied on bounds like the multiparameter Cram\'er-Rao bound and its quantum extension by Helstrom \cite{helstrom1976}. While powerful, these bounds have limitations, particularly when faced with practical constraints such as limited measurement data and/or moderate prior knowledge.

Recent research has explored alternatives, including the Holevo bound, which offers more information but still requires locally unbiased estimators or infinite numbers of experimental repetitions to make it sense \cite{holevobook}. To overcome these difficulty, Conlon {\it et al.} \cite{conlon2021efficient} proposed a new approach which is applicable for locally unbiased and  finite sample setting. This new bound was named after the historical contribution as the Nagaoka-Hayashi bound. To address these challenges in practical situation, the Bayesian framework emerges as a promising approach due to its ability to incorporate prior information effectively without assumption of unbiasedness.  This was initiated by Personick in his seminal work and it established optimal lower bound for one-parameter case \cite{personick71}. Since Personick's work, many attempts were made to extend it to multiparameter estimation. However, current solutions still face many challenges, in particular when the Bayes risk is defined by parameter-dependent weight matrix. Tsang \cite{tsang2020physics} proposed a quantum version of the Gill-Levit bound \cite{GillLevit1995} to incorporate parameter-dependent weight matrix. The Tsang bound, however, is expressed as non-trivial optimization. In a recent work, the Bayesian version of the Nagaoka-Hayashi bound was proposed in terms of a semidefinite programming problem, which is efficiently computable \cite{suzuki2024bayesian}. 

In Suzuki's work \cite{suzuki2024bayesian}, it was also attempted to derive lower bounds which is easy to evaluate than the Bayesian Nagaoka-Hayshi bound, but only a partial solution was given. In this paper, we aim at deriving the more explicit form of the Bayesian lower bound for weighted trace of the mean squared error (MSE) with parameter-independent weight matrices, which is expressed only in terms of averaged quantities over the prior information. This lower bound (Theorem \ref{thm:BH2}) is regarded as a Bayesian version of the Holevo bound in point estimation, since it is given as optimization of a function of the set of Hermitian operators. Based on our result, the close form of one parameter family in parameter-independent weight matrix follows up immediately from our result (Theorem \ref{thm:BH3}). It is shown that the latter bound is a generalization of previously known Bayesian lower bounds, since it is reduced to the Personick bound and its variants as special cases.


\section{Preliminaries}

\subsection{Optimal Bayesian MSE for classical estimation}\label{seq:optclassical}
In this section, we recall the Bayesian classical case at first. The notations referring to the Bayesian statistics are listed as bellows: 
\begin{itemize}
    \item $p_\theta(x)=p(x|\theta)$: A parametric model (likelihood) = A family of probability distribution on $\mathcal{X}$.
    \item $\pi(\theta)$: A prior distribution for the parameter space $\Theta$.
    \item $q(\theta|x)$: A posterior distribution for a datum $x\in\cX$.
    \item $p_x(x)$: A marginal distribution (evidence).
\end{itemize}
The expectation value of a random variable $X$ with respect to the model $p_\theta$ is denoted by $E_\theta[X]=\int_x x\, p_\theta(x) dx$, whereas the expectation value with respect to the joint distribution $\pi p_\theta$ is expressed without subscript, and the expectation value with respect to the marginal $p_X$ is denoted as $E_X[X]$.  
An estimator that returns values on the set $\Theta$ is denoted as $\hattheta(x)=(\hattheta_1(x),\hattheta_2(x),\cdots,\hattheta_n(x) )$. 
The main objective of this proceedings is to minimize the average of the weighted trace of the MSE matrix. 
We consider the most general problem with a parameter-dependent weight matrix, which is $n\times n$ positive definite matrix: $W(\theta)=[W_{jk} (\theta)] > 0$. The Bayesian risk is then defined by
\begin{equation}
    R_\mathrm{B} [\hattheta|W] = \int d\theta \pi(\theta) \Tr \left[ W(\theta) V_\theta [ \hattheta] \right],
\end{equation}
where $V_{\theta} [\hattheta]$ is the MSE matrix with the $j,k$ elements 
\begin{equation}
    V_{\theta,jk} [\hattheta] = E_\theta \left[ (\hattheta_j(X)-\theta_j) (\hattheta_k(X)-\theta_k) \right].
\end{equation}
Throughout the paper, $\Tr[\cdot ]$ is the trace over $n$-dimensional parameter space. 

The optimal estimator is known to be given as follows. 
First define an inner product on the space of $n$-dimensional vectors of random variables $F(X)=(f_1(X),f_2(X),\ldots,f_n(X))$ and $G=(g_1(X),g_2(X),\ldots,g_n(X))$ by 
$\braket{F}{G}:=\sum_{j=1}^nE_X[f_j(X)g_j(X)]$. The Bayes risk is minimized by completing the square as
\begin{equation*}
\begin{split}
    R_\mathrm{B}[\hattheta|W] &= \braket{\hattheta}{{W}^\post\hattheta} - 2 \braket{{d}^\post}{\hattheta}
    +\ol{w} \\
    & \geq \braket{{d}^\post}{({W}^\post)^\inv {d}^\post}+\ol{w},\label{def:const}
\end{split}
\end{equation*}
where the first equality follows by defining the following quantities. 
\begin{equation}
\begin{split}
    {W}^\post_{ij}(x)&:=\int d\theta q(\theta|x)W_{ij}(\theta), \\ 
    {d}^\post_{j}(x)&:= \sum_{i=1}^n \int d\theta q(\theta|x)W_{ij}(\theta)\theta_i,\\
    \ol{w} &:=\int \,d\theta\,\prior \sum_{i,j=1}^n\theta_j W_{ij}(\theta)\theta_j.
\end{split}
\end{equation}
Furthermore, the optimal estimator is identified as
\begin{equation}
    \hattheta_{\mathrm{opt},i}(x) = \sum_{j=1}^n\left( ({W}^\post(x))^{-1}\right)_{ij} {d}^\post_{j}(x). 
\end{equation}
However, this result cannot be extended in the quantum case directly because a quantum version of ${d}^\post$ is not uniquely defined due to non-commutativity.

\subsection{Quantum Bayesian estimation}
We aim at deriving lower bounds in another direction. Before that, we need to define the problem in quantum Bayesian setting. Let $\Hil$ be a finite dimensional Hilbert space and $\left\{ S_\theta~|~ \theta \in \Theta \right\}$ be an $n$-parameter quantum model with $\theta=(\theta_1, \theta_2, \cdots, \theta_n)$. In the following, we consider $S_\theta$ is full rank for all $\theta \in \Theta$ and the map $\theta \rightarrow S_\theta$ is one-to-one. Let a set of positive semidefinite matrices $\{\Pi_x\}$ be a measurement with the measurement outcome $x\in\cX$. The quantum measurement is normally referred to as a positive operator-valued measure (POVM), which is defined by
\begin{equation}
    \Pi=\{ \Pi_x \},~ \forall x \in \cX,~ \Pi_x \geq 0,~ \int_\cX dx \Pi_x = I,
\end{equation}
where $I$ is the identity operator on $\Hil$. 

Measurement outcome is described by a random variable $X$ that obeys the conditional probability distribution: 
\[
p_\theta(x) = \tr \left[ S_\theta \Pi_x \right],
\] 
where $\tr[ \cdot]$ denotes the trace on the Hilbert space $\Hil$. 
The performance of the estimator in this study is quantified by 
the quantum Bayes risk:
\begin{equation}\label{def:Brisk2}
  R_\mathrm{B} [\Pi,\htheta]= \int d\theta\,\prior \Tr\left[ W(\theta) V_\theta[\Pi,\htheta]\right].
\end{equation}
The main objective is to find the best pair $(\Pi, \hattheta)$ that minimizes the Bayes risk.

\subsection{Bayesian Nagaoka-Hayashi bound}
In a recent work, we have proposed a new class of Bayesian lower bound, called the Bayesian Nagaoka-Hayashi bound \cite{suzuki2024bayesian}. 
Before we recall this bound, we obtain the alternative form of the Bayes risk. 
We first introduce a new set of variables for quantum measurement and estimator by
\begin{align*}
\mathbb{L}_{jk}[\Pi,\hat{\theta}]&=\int dx \hat{\theta}_{j}(x)\Pi_x\hat{\theta}_{k}(x)\quad(j,k=1,2,\dots,n),\\
X_j[\Pi,\hat{\theta}]&=\int dx\hat{\theta}_{j}(x)\Pi_x \quad(j=1,2,\ldots,n).
\end{align*}
In the following, we will omit the argument $[\Pi,\htheta]$ when it is clear. 
The key idea is to regard the above quantities as the operator-valued matrix and vector \cite{hayashi99,conlon2021efficient}. 
We introduce the extended Hilbert space by $\bbH=\bbc^n\otimes\cH$, and then $\bbL$ is identified as a matrix on $\bbH$, whose block matrix representation is given by $[\mathbb{L}_{jk}]$. 
 
Next, define the following quantities defined by the model and prior distribution. 
\begin{align*}
    \bbS(\theta)&:=[\bbS_{ij}(\theta)]\ \mbox{with}\ 
    \bbS_{ij}(\theta):= W_{ij}(\theta) S_\theta ,\\
    \Mtheta&:=[D_i(\theta)]\ \mbox{with}\ 
    D_i(\theta):=\sum_{j=1}^n W_{ij}(\theta)\theta_j S_\theta. 
\end{align*}
Then the next lemma is immediate. 
\begin{lemma}[Suzuki \cite{suzuki2024bayesian} Lemma 2]\label{lem:Briskrep}
The Bayes risk is expressed as
\begin{align}\nonumber
 R_\mathrm{B}[\Pi,\htheta]&=\Ttr{\left[\ol{\Sbb} \bbL\right]}-\Ttr{\left[\Mbar X^{T_1}\right]}-\Ttr{\left[X \Mbar^{T_1}\right]}+ \ol{w} ,\\
\mathrm{where\ }\bbSbar&:=\int \,d\theta\,\prior\bbS(\theta),\label{def:bayesS}\\
\Mbar&:=\int \,d\theta\,\prior\Mtheta .\label{def:bayesT}
\end{align}
\end{lemma}
In this lemma, $\overline{\cdot}$ denotes the averaged operators with respect to the prior distribution. 
$X^{T_1}$ means the transpose over the parameter space and $\Ttr[\cdot] $ is the trace over both the Hilbert space and the parameter space.

The fundamental theorem of Ref.~\cite{suzuki2024bayesian} is the lower bound for the Bayes risk $ R_\mathrm{B}[\Pi,\hattheta]$.
\begin{theorem}[Bayesian Nagaoka-Hayashi bound, Suzuki \cite{suzuki2024bayesian} Theorem 1]\label{thm:BNHbound}
For any POVM $\Pi$ and estimator $\htheta$, the following lower bound holds for the Bayes risk. 
  \begin{align*}
& R_\mathrm{B} [\Pi,\hat{\theta}]\ge \cNH \nonumber\\
  &\cNH:=\min_{\bbL,X}\left\{ \Ttr{\left[\bbSbar\bbL\right]}
  -\Ttr{\left[\Mbar X^{T_1}\right]}-\Ttr{\left[X \Mbar^{T_1}\right]}\right\}+ \sfwbar .
  \end{align*}
Here optimization is subject to the constraints:
  $\forall j,k,\bbL_{jk}=\bbL_{kj}$, $\bbL_{jk}$ is Hermitian, $X_j$ is Hermitian, and $\bbL\geq {X} X^{T_1}$.
\end{theorem}	

\subsection{Bayesian Holevo-type bound}
In Ref.~\cite{suzuki2024bayesian}, a Bayesian version of the Holevo bound was also proposed. 
We recall that the Bayesian Holevo-type bound is lower than the Bayesian Nagaoka-Hayashi bound in general.
\begin{theorem}[Bayesian Holevo-type bound, Suzuki \cite{suzuki2024bayesian} Theorem 3]\label{thm:BHbound}
\begin{align*}
    \cNH& \geq \Ccal_\mathrm{BH},\\
    \Ccal_\mathrm{BH}&=\min_{X_j:\mathrm{\,Hermitian}} \{ \int d\theta \pi(\ta)\Tr \left[\Real Z_\theta(X)\right] \\
    + \int d\theta &\pi(\ta) \Tr|\Imag Z_\theta (X)| 
 - \Ttr\left[ \ol{D} X^{T_1}\right] -\Ttr{\left[X \Mbar^{T_1}\right]}\}+\ol{w},
\end{align*}
where 
\begin{equation*}
    Z_\theta(X) = \tr \left( \sqrt{W(\theta)} \otimes \sqrt{S_\theta} X X^{T_1} \sqrt{W(\theta)} \otimes \sqrt{S_\theta} \right).
\end{equation*}
Here $\Re(\Im)$ denote the element-wise real (imaginary) part of a matrix and 
$|A|=\sqrt{A^\dag A}$. In other words, $\Tr|A|$ is the absolute sum of all eigenvalues of $A$.
\end{theorem}

This lower bound is only expressed in terms of minimization over the set of Hermitian matrices $X=(X_1,X_2,\ldots,X_n)$. 
However, it is still hard to analyze it further, since the second term is expressed as an integration of functions containing variable $X$. 

\section{Results}
In this section, we shall propose a method of deriving a family of Holevo-type lower bounds based on a new technique. 
This is to apply an appropriate map that satisfies a certain conditions. 
Before we present our result, we first examine the direct analogue of the optimal classical bound explained in Sec.~\ref{seq:optclassical}. 

\subsection{Direct generalization of the Bayesian classically optimum bound}
Firstly, we consider to generalize the Bayesian classically optimum bound in quantum state with $\ta$-dependent weight matrix $W(\ta)$. 
This is done by apply the constraint $\bbL\geq {X} X^{T_1}$ in the definition of the Bayesian Nagaoka-Hayashi bound. 
\begin{theorem}[Quantum Bayesian lower bound 1]\label{thm:QBbound_direct}
    \begin{equation}
    R_\sfB \geq \Ccal_\mathrm{direct} = - \braket{\ol{D}}{\ol{\Sbb}^{-1} \ol{D}}_{\ol{\Sbb}} +\ol{w}, 
\end{equation}
where the inner product is defined for operator-valued vector $X=[X_j]$ and $Y=[Y_j]$ as 
\begin{equation}
    \inner{X}{Y}_{\ol{\Sbb}}:= \tr\left[X^{T_1} \ol{\Sbb} Y\right]. 
\end{equation}
\end{theorem}

\begin{proof}
This can be proved by completing the square as in the classical case after the use of $\bbL\geq {X} X^{T_1}$.
\begin{equation} \label{proof:direct}
    \begin{split}
        R_\sfB & \geq \Ttr[ \overline{\Sbb} X X^{T_1}]
        - 2  \Ttr \left( \overline{D} X^{T_1} \right)  +\ol{w} \\
        &= \tr \left( X^{T_1} \overline{\Sbb} X  \right)
        - 2 \tr \left[ X^{T_1} \overline{D}   \right] + \ol{w} \\
        &= \braket{X}{ X}_{\ol{\Sbb}} -2 \braket{\ol{\Sbb}^{-1} \ol{D}}{ X }_{\ol{\Sbb}}+\ol{w} \\
        &\geq - \braket{\ol{D}}{\ol{\Sbb}^{-1} \ol{D}}_{\ol{\Sbb}} +\ol{w},
    \end{split}
\end{equation}
where the optimizer for the inequality is given by 
\[
X_{\mathrm{opt},i}=\sum_{j=1}^n ({\ol{\Sbb}}^{-1})_{ij} \overline{D_j}. 
\]
\end{proof}

The proof of this bound is straightforward. However,
this cannot be tight because $X_j, ({\ol{\Sbb}}^{-1})_{ij}$ ($i,j$th element of $(\ol{\Sbb})^{-1}$), and $\overline{D_i}$ are Hermitian, but $\ol{\Sbb}^{-1}_{ij}$ and $\overline{D_i}$ are not commute. This means that optimal $X_{\mathrm{opt}}$ cannot be Hermitian. 
This bound is the direct generalization of the Bayesian classical one. If we set the state as classical we can immediately degenerate in the previous bound.

\subsection{New quantum Bayesian lower bounds}
Next, we derive a new one-parameter family of lower bounds which is tighter than the direct generalization.
The key idea in this proceedings is to define a map $\E$ from $\Li(\C^n \ot \Hil)$ to $\Li(\C^n)$ such 
that satisfies the condition. 
\begin{definition}
We say that $\E$ is real symmetric, if $\E(\bbL)\in \bbc^{n\times n}$ is real symmetric for all $\bbL\in\Li_\mathrm{sym}(\C^n \ot \Hil)$. 
\end{definition}
Here $ \Li_\mathrm{sym}(\C^n \ot \Hil)$ denotes the set of all Hermitian $\bbL$ which is invariant under the partial transposition $\cdot^{T_1}$. Using the block-matrix representation, this is equivalent to $\forall i,j,\,\bbL_{ij}=\bbL_{ji}$. 

Given a positive operator $\bbS$ on $\bbH=\C^n \ot \Hil$, we define one-parameter family of maps for $\Xds \in \Li(\C^n \ot \Hil)$.
\begin{equation}
    \E^{(\la)}(\Xds) := \frac{1+\lambda}{2} \E^{(+)}(\Xds) + \frac{1-\lambda}{2} \E^{(-)}(\Xds),
\end{equation}
where
\begin{align}
    \E^{(+)}(\Xds) &:= \tr \left[ (I\ot \sq{S}) \Xds (I\ot \sq{S}) \right]\\
    \E^{(-)}(\Xds) &:= \tr \left[ (I\ot \sq{S}) \Xds^{T_1} (I\ot \sq{S}) \right]
\end{align}

We can show the next property. 
\begin{lemma}\label{lem:lambda-map}
A map $\E^{(\la)}$ is real symmetric and positive for $\lambda\in[-1,1]$. 
\end{lemma}
\begin{proof}
To show $\E^{(+)}$ is real symmetry we decompose $\Xds=\sum_i \alpha_i A_i \ot B_i$. Then we obtain the form of $\E^{(+)} (\Xds)$ in 
    \begin{align}
        \E^{(+)} (\Xds) = \sum_i \alpha_i A_i \tr S B_i \ \mathrm{and} \  
        \E^{(+)} (\Xds)^T = \sum_i \alpha_i A_i^T \tr S B_i.
    \end{align}
    We notice that $\Xds \in \Li_{sym,h}(\Hil)$, this means $A_i^T = A_i$. Thus $\E^{(+)} (\Xds)=\E^{(+)} (\Xds)^T$. 
    
    On the other hand, $\E^{(+)} (\Xds)$ is a real matrix. Since $A$ is positive and $A^T =A$, we obtain $A \in \R^{n \times n} $. By $B$ is hermitian, we get $\tr SB \in \R$. Thus $\E^{(+)} (\Xds)\in \R^{n\times n}$. 
\end{proof}

After identifying a real symmetric map, we can derive a family of lower bounds, which has a similar structure as the Holevo bound in point estimation.  
\begin{theorem}[Quantum Bayesian lower bound 2]\label{thm:BH2}
\begin{equation} \label{bound:BH2}
\begin{split}
    \Ccal_\mathrm{BNH} \geq \Ccal_\mathrm{BH}^{(\la)}:= \min_X \{ \Tr \left[ W \Real Z_B^{(\la)} (X)\right] \\
    + \Tr | \sqrt{W} \Im Z_B^{(\la)}(X) \sqrt{W} | \\ - \Ttr\left[ W \ol{D} X^{T_1}\right] -\Ttr{\left[ W X \Mbar^{T_1}\right]}+ \Tr[W M] \},
\end{split}
\end{equation}
where
\begin{align}
    Z_B^{(\la)}(X)&:= \left[ \frac{1+\la}{2} \tr(X_i S_B X_j) +\frac{1-\la}{2} \tr (X_j S_B X_i) \right],\\
    &= \left[ \inner{X_i}{X_j}^{(\la)}_{S_B} \right]
\end{align}
is an $n\times n$ complex positive semidefinite matrix, 
and minimization is subject to $X_j$: Hermitian. 
Furthermore, the right hand side the bound \eqref{bound:BH2} takes the maximum value when $\lambda=\pm1$.
\[
\max_{\lambda\in[-1,1]}=\Ccal_\mathrm{BH}^{(\la= \pm 1)} .
\]
\end{theorem}

\begin{proof}
First, let us prove the inequality \eqref{bound:BH2}. 
From the original form of the Bayesian Nagaoka-Hayashi bound of Theorem \ref{thm:BNHbound}, 
we see that $\bbL,X$ satisfying the constraints also satisfy $\E^{(\la)}(\bbL-X X^{T_1})\ge0$. 
Define $V:=\E^{(\la)}(\bbL)=\Ttr{\left[\ol{\Sbb} \bbL\right]}$, then $V$ is real symmetric and hence we can relax the minimization problem in terms of 
$V$ and $X$: 
\begin{align}\label{eq:key_eq}
 \Ccal_\mathrm{BNH} &\geq \min_{V,X}\{\Tr[WV]- \Ttr\left[ W\ol{D} X^{T_1}\right] -\Ttr\left[  W X \ol{D}^{T_1}\right]\\& + \Tr[WM] \}, 
\end{align}
where the constraints are $V$: real symmetric, $V\ge Z^{(\la)}(X)$, and $X_j$: Hermitian. 
From the well-known lemma (Lemma 6.6.1 in Ref.~\cite{holevobook}), we 
have $\min_V\{\Tr[WV]\,|\,V\ge Z^{(\la)}(X)\}= \Tr \left[ W \Real Z^{(\la)} (X)\right] + \Tr |\sq{W} \Im Z^{(\la)}(X) \sq{W}|$. 
This proves the first part of the theorem.

Next, we show that this family of bounds takes the maximum value when $\lambda=\pm 1$, i.e.,
\begin{equation}
    \Ccal_\mathrm{BH}^{{(\la=\pm1)} } \geq \Ccal_\mathrm{BH}^{(\la)}~ \rm{for~all~}\la \in [-1,1]. 
\end{equation}
By separating the items, we see the real part of $Z^{(\la)}$ is independent of $\lambda$ as 
\begin{equation}\label{eq:RePartZ}
\begin{split}
    \Real Z^{(\la)}(X)&= \half Z^{(\la)} (X)+ \half (Z^{(\la)}(X)^T)\\
    &= \half Z^{(\la)} (X)+\half Z^{(-\la)} (X)\\
    &=\frac{1}{2} \tr\left[\left\{\ol{\bbS}\,,\, \half \left[ X X^{T_1}+(X X^{T_1})^{T_1}\right] \right\} \right], 
    \end{split}
\end{equation}
where the second line is due to $(\tr[\ol{\bbS}\mathbb{X}])^{T}=\tr[\mathbb{X}^{T_1}\ol{\bbS}]$ and symmetrized multiplication of $\bar{\bbS},\mathbb{X}$. 
The imaginary part is, on the other hand, proportional to $\la$, since 
\begin{equation}\label{eq:ImPartZ}
    \begin{split}
        \Imag Z^{(\la)}(X) &= \frac{1}{2 i} \left( Z^{(\la)}(X) - Z^{(\la)}(X)^T \right)\\
        &= \frac{1}{2 i} ( Z^{(\la)} (X)- Z^{(-\la)} (X) )\\
        &= \frac{\la}{2}\tr\left[\left\{\ol{\bbS}\,,\,\frac{i}{2} \big[ X X^{T_1}-(X X^{T_1})^{T_1} \big]\right\} \right],
    \end{split}
\end{equation}
where $\{X,Y\}:=XY+YX$ is the anti-commutator. 
Due to it, we can see the trace norm of imaginary part is proportional to absolute value of $\la$. In other words, $\Tr|\Imag Z^{(\la)}(X)| \propto | \la|$. This means that the one-parameter family of the quantum Bayesian lower bounds takes the maximum value at $\la=\pm 1$.
\end{proof}

We now discuss the advantages of the proposed bound:
\begin{itemize}
    \item The two Bayesian Holevo-type bounds are similar; both are expressed as minimization over $X$ operators. However as an optimization problem, the new one is easier to calculate since the previous bound (Theorem \ref{thm:BHbound}) requires to compute the integral over a complicated function of $X$, namely $\int d\theta \pi(\ta)  \Tr|\Imag Z_\theta (X)|$. 
    In contrast, the new bound refers to the averaged state $\ol{\bbS}=\int d\theta W(\theta)\otimes S_\theta$ only. Thus, what we need is to calculate the averaged state in advance.
    \item 
    This new bound is tighter than the direct quantum version of the optimal Bayesian bound (Theorem \ref{thm:QBbound_direct}).
    \begin{equation}
        \Ccal_\mathrm{BH}^{(\la)} \geq  \Ccal_\mathrm{direct}.
    \end{equation}
    The reason is that we can show the real part of the bound \eqref{eq:RePartZ} is equal to the first term of function in the proof \eqref{proof:direct} and the imaginary part \eqref{eq:ImPartZ} is non-negative.
    \begin{multline*}
        \Tr [\Real Z^{(\la)}(X)]+\Tr | \Imag Z^{(\la)}(X)| \\
        \geq \Tr [\Real Z^{(\la)}(X)] = \tr [X^{T_1} \ol{\Sbb} X] .
    \end{multline*}
    \item Lastly, the family of our new bounds immediately gives one-parameter family generalization of the previously established Bayesian lower bounds when the weight matrix is parameter independent. This will be explored in the next subsection. 
\end{itemize}

\subsection{Bayesian $\lambda$LD bound in parameter-independent weight}
In this subsection, we take the $\ta$-independent weight matrix. The Bayesian MSE is defined as
\begin{equation}
    V_B[\Pi,\hattheta]:= \int d\ta \pi(\ta) V_\ta [\Pi,\hattheta].
\end{equation}
The goal here is to derive a family of quantum Bayesian lower bounds in a closed form by approximating 
the optimization over $X$ operators. 
The key ingredient in our approach is to utilize the logarithmic derivative (LD) type equation which 
was shown to be useful for point estimation \cite{yamagata2021maximum,suzuki2021non}. 
For convenience, we call this family as the Bayesian $\la$LD bound, even though this is different from the quantum Cram\'er-Rao  bound based on quantum LDs. 

We substitute parameter-independent weight matrix in the quantum Bayesian bound of Theorem \ref{thm:BH2} to get the following theorem. 
\begin{theorem}[Bayesian $\lambda$LD bound] \label{thm:BH3}
For parameter independent weight matrix, the Bayes risk is bounded by $\Ccal^{(\la)}_\mathrm{BLD} (W)$. 
\begin{equation*}
    \Ccal^{(\la)}_\mathrm{BLD} (W) :=  -\Tr \left[W \Real K^{(\la)}\right] + \Tr | \sq{W} \Imag K^{(\la)} \sq{W} |+ \ol{w},
\end{equation*}
    for $\lambda\in[-1,1]$, where $K^{(\la)}$ is defined by
\begin{align*}
    K^{(\la)}_{jk}&:= \tr[\MBk L^{(\la)}_j], \\
    \MBj &= \plam \rhoB L^{(\la)}_j + \mlam L^{(\la)}_j \rhoB,\\
\rhoB&=\int \,d\theta\,\prior S_\theta,\\
\MBj&=\int \,d\theta\,\prior \theta_j S_\theta. 
\end{align*}
In addition, the Bayes risk is bounded by 
\begin{equation}
    R_{\mathrm{B}}[\Pi,\hattheta] \geq \max_\la 
    \Ccal^{(\la)}_\mathrm{BLD} (W)=: C^{\max}_\mathrm{BLD}.
\end{equation}
\end{theorem}
We have two remarks before proving this theorem. 
First, note that the last form of the bound is a Bayesian version of the maximum logarithmic derivative Cram\'e-Rao bound proposed by Yamagata in point estimation \cite{yamagata2021maximum}. 
Second, due to the process of deriving the Bayesian $\la$LD bound, it is tighter than the Personick bound \cite{personick71} and its generalization \cite{rubio2019quantum,demkowicz2020multi}, and the Bayesian bound proposed by Holevo \cite{holevo_qest,holevo1977commutation}. In other words, the above one-parameter family of bounds includes well-known bounds as special cases. This will be shown later. 

\begin{proof}
When $W$ is $\theta$-independent, we have
\begin{align*}
\bbSbar&=W\otimes \rhoB,\\
 \Mbarj&=\sum_kW_{jk}\MBk,\\
\sfwbar &=\sfTr{WM}, 
\end{align*}
where $M_{jk}:=\int d\theta\,\prior\theta_j\theta_k$.  
Define one-parameter family of inner products by 
\begin{equation}
\lamin{X}{Y}:=\tr\left[\rhoB (\plam YX^\dagger+ \mlam X^\dagger Y)\right],  
\end{equation}
for $\lambda\in[-1,1]$.
Then the Bayesian bound of Theorem \ref{thm:BH2} is expressed by using Eq.~\eqref{eq:key_eq} as
\[
\Ccal_\mathrm{BH}^{(\la)}=\min_{V,X}\{ \Tr[W V]+\ol{w}\,|\,V\ge \sfZB[X]-\sfHB[X]-\sfHB[X]^T\}, 
\]
where $Z^{(\la)}_{\mathrm{B},jk}[X]:=\lamin{X_j}{X_k}$, $H_{\mathrm{B},jk}[X]:=\tr{[\MBj X_k]}$, and 
$V$ is real symmetric. We can carry out the optimization over $X$ first to get a lower bound. 
To do this, we note the constraint is also written as 
\[
\sfZB[X]-\sfHB[X]-\sfHB[X]^T = \sfZB[X-L^{(\la)}]-\sfZB[L^{(\la)}], 
\]
which follows from the definition of $L^{(\la)}_j$. 
Note that $L^{(\la)}_j$ are not Hermitian in general, yet we can obtain a lower bound by removing 
this condition on $X$. Hence we have
\begin{equation}
\Ccal_\mathrm{BH}^{(\la)}\ge\min_{V}\{ \Tr[W V]+\ol{w}\,|\,V\ge -\sfZB[L^{(\la)}]\}, 
\end{equation}
where $V$ real symmetric. The last step follows from the Lemma (Lemma 6.6.1 in Ref.~\cite{holevobook}) to get
\begin{align*}
\Ccal_\mathrm{BH}^{(\la)}&\ge -\Tr[W\Re\sfZB[L^{(\la)}]]+\Tr|\sqrt{W}\Im \sfZB[L^{(\la)}]\sqrt{W}|+\ol{w}\\
&=\Ccal^{(\la)}_\mathrm{BLD} (W), 
\end{align*}
since $\sfZB[L^{(\la)}]=K^{(\la)}$. 
\end{proof}

By the above proof, we can immediately show that the Personick bound is recovered when $\lambda=0$, 
whereas the Holevo's result corresponds to $\lambda=1$. To do this, we can set $\theta$-independent $W$ as the rank-1 projector $W=c c^{\dagger}$ with $c \in\C^n$.
  Next, we obtain a lower bound for $\Ccal_\mathrm{BLD}^{(\la)}$.
  This gives the desired result.
\begin{align*}
  \Ccal_\mathrm{BLD}^{(\la)}-c^\dagger M c
  &\ge -\lamin{ L_c^{(\la)}}{L_c^{(\la)}}
   =-c^\dagger K^{(\la)} c,
\end{align*}
where we set $L_c^{(\la)}=\sum_j c_j L_j^{(\la)}$.  
Since this is true for any choice of $c$, we prove the matrix inequality $V_\mathrm{B} \ge M- K^{(\la)}$. If we choose $\la=0$, thereby, we recover the Personick bound \cite{personick71}.  
%

\section{Conclusion}
In summary, we introduce one-parameter family of the Bayesian Holevo type bounds, a new lower bound for the Bayes risk, which addresses the parameter-independent weight matrices in this domain. Our bound exhibits superior tightness compared to the direct extension of the Bayesian classical optimum bound, while maintaining computational simplicity relative to the Bayesian Holevo-type bound. 
Moreover, we derive the closed form of the Bayesian $\lambda$LD type lower bounds. Our analysis demonstrates its superiority over existing lower bounds, including the symmetric logarithmic derivative type bound \cite{personick71} and the right logarithmic derivative type bound \cite{holevo_qest,holevo1977commutation}, which are obtained by setting $\lambda$ to 0 and 1, respectively. In other words, the derived closed form offers a generalized framework that encompasses other bounds, enhancing its practical utility and theoretical significance.

\bibliographystyle{IEEEtran}
\bibliography{ref}

\vspace{12pt}
\color{red}

\end{document}